\documentclass[submission,copyright,creativecommons]{eptcs}
\providecommand{\event}{} 
\usepackage{breakurl}             
\usepackage{underscore}           
\usepackage{mathpartir}
\usepackage{amsmath}
\usepackage{mathtools,amssymb}
\usepackage [utf8] {inputenc}
\usepackage[T1]{fontenc}  
\usepackage{xcolor}
\usepackage{enumerate}
\usepackage{prooftree}
\definecolor{blue-violet}{rgb}{0.54, 0.17, 0.89}

 \newtheorem{theorem}{Theorem}[section]
\newtheorem{definition}[theorem]{Definition}
 
 \newtheorem{lemma}[theorem]{Lemma}
 \newtheorem{example}[theorem]{Example}
 
\newenvironment{proof}{{\em Proof.}}{\hfill$\square$\medskip}

\newcommand{\depthk}[3]{\ensuremath{\delta_{{\sf sem}}^{#1}}(#2,#3)}
\newcommand{\Nat}{\mathbb{N}}

\usepackage{relsize,xspace}

\newcommand{\cau}[2]{#1\circ#2}

\newcommand{\inp}[5]{\SigmaB_{#2\in#3}#1?\Seq{#4_#2}{#5_#2}}

\newcommand{\oup}[5]{\bigoplus_{#2\in#3}#1!\Seq{#4_#2}{#5_#2}}

\newcommand{\oupTx}[5]{\textstyle{\bigoplus}_{#2\in#3}#1!\Seq{#4_#2}{#5_#2}}

\newcommand{\gt}[6]{#1\to#2: \GlSyB_{#3\in#4}\Seq{#5_#3}{#6_#3}}

\newcommand{\gtBin}[4]{#1\to#2: (#3\,\GlSy\, #4)}
\renewcommand{\gt}[6]{#1\to#2: \{ \Seq{#5_#3}{#6_#3} \}_{#3\in#4}}   
\renewcommand{\gtBin}[4]{#1\to#2: \{#3\, , \, #4\}}   

 \newcommand{\RG}{\ensuremath{{\sf G}}}
 
 \newcommand{\PiB}{\mathlarger{\mathlarger\Pi}}
  \newcommand{\SigmaB}{\mathlarger{\mathlarger\Sigma}}

  \newcommand{\EGG}{\mathcal{E}}

  \newcommand{\ee}{\epsilon}

  \newcommand{\depthG}[2]{\delta(#1,  #2)}
  
  \newcommand{\gdepthS}[2]{\delta_{{\sf gsem}}(#1,  #2)}
   \newcommand{\gdepthkS}[2]{\delta^k_{{\sf gsem}}(#1,  #2)}
 \newcommand{\mylabel}[1]{\label{#1}}

\newcommand{\comocc}{\ensuremath{\gamma}}

\newcommand{\concat}[2]{\ensuremath{#1\,{\cdot}\,#2}}

 \newcommand{\inpP}[5]{\SigmaB_{#2\in#3}#1?\Seq{#4_#2}{#5}}
\newcommand{\oupP}[5]{\bigoplus_{#2\in#3}#1!\Seq{#4_#2}{#5}}

\newcommand{\refToFigure}[1]{Figure~\ref{#1}}
\newcommand{\refToSection}[1]{Section~\ref{#1}}
\newcommand{\refToExample}[1]{Example~\ref{#1}}
\newcommand{\refToLemma}[1]{Lemma~\ref{#1}}
\newcommand{\refToTheorem}[1]{Theorem~\ref{#1}}

\newcommand{\refToDef}[1]{Definition~\ref{#1}}

\newcommand{\set}[1]{\{#1\}}
\newcommand{\kf}[1]{\ensuremath{\mathsf{#1}\xspace}}
 
\newcommand{\cSinferrule}[3][]{
  \mprset{fraction={===},
  fractionaboveskip=0.2ex,
  fractionbelowskip=1.30ex}
  \inferrule[#1]{#2}{~#3}
}

\newcommand{\PP}{\ensuremath{P}}

\newcommand{\la}{\lambda}

\newcommand{\M}{\lambda}

\newcommand{\pp}{{\sf p}}
\newcommand{\q}{{\sf q}}
\newcommand{\pr}{{\sf r}}
\newcommand{\ps}{{\sf s}}
\newcommand{\pt}{{\sf t}}

\newcommand{\sendL}[2]{#1!#2}
\newcommand{\rcvL}[2]{#1?#2}

\newcommand{\pc}{~|~}
\newcommand{\Seq}[2]{#1;#2}

\newcommand{\inact}{\ensuremath{\mathbf{0}}}

\newcommand{\Q}{\ensuremath{Q}}
\newcommand{\R}{\ensuremath{R}}

\newcommand{\del}[1]{\ensuremath{\delta}}
\newcommand{\Nt}{\ensuremath{{\sf N}}}

\newcommand{\parN}{\mathrel{\|}}
\newcommand{\pP}[2] {#1[\![\,#2\,]\!]}

\newcommand{\stackred}[1]{\xrightarrow{#1}}

 \newcommand{\G}{\ensuremath{{\sf G}}}

 \newcommand{\End}{\kf{End}}

\newcommand{\partcomm}[1]{\ensuremath{{\sf part}(#1)}}

 \newcommand{\participant}[1]{\partcomm{#1}}

\newcommand{\proj}[2]{#1 \!  \upharpoonright  \! #2\,}

\newcommand{\subt}{\leq}

\newcommand{\derN}[2]{\vdash #1 :#2}

\newcommand{\rulename}[1]{{[\textsc{#1}]}}

\newcommand{\gtCom}[3]{#1\stackrel{#3}{\to}#2}

\newcommand{\gr}{\ensuremath{~\#~}}

\newcommand{\grr}{\ensuremath{\,\#\,}}
\newcommand{\grin}{\ensuremath{\,\#_{in}\,}}

\newcommand{\impl}{\ensuremath{\Rightarrow}}

\newcommand{\precP}{\ensuremath{\leq}}

 \newcommand{\GEs}{\mathcal{G\!E}}

\newcommand{\comseq}{ \sigma } 

\newcommand{\Comseq}{\textit{Traces}}

\newcommand{\eqclass}[1]{\ensuremath{[#1]_{\sim}}}

\newcommand{\quotient}{\ensuremath{\Comseq/\!\!\sim}}

\newcommand{\ESet}{\ensuremath{\mathcal{X}}}

\newcommand{\ESG}[1]{\ensuremath{\mathcal{S}(#1)}}

\newcommand{\emptyseq}{\ensuremath{\epsilon}}

\newcommand{\comm}[1]{\ensuremath{{\sf cm}}(#1)}

\newcommand{\last}[1]{\ensuremath{{\sf last}}(#1)}

\newcommand{\coDefGr}{::=^{coind}}

\newcommand{\GP}{{\G}}

\newcommand{\FPaths}[1]{{\sf Tr^+}(#1)}

\newcommand{\Comm}[3]{#1#3#2}

\newcommand{\ev}[1]{\ensuremath{{\sf ev}(#1)}}

 \newcommand{\projS}[2]{#1 @   #2\,}

\newcommand{\at}[2]{#1[#2]}

\newcommand{\range}[3]{#1[#2\,...\,#3]}

\newcommand{\cardin}[1]{\!\!\pc\! #1\!\pc\!}

\title{Towards a Semantic Characterisation \\
  of Global Type Well-formedness
}

\author{Ilaria Castellani\footnote{This research has been supported by the ANR17-CE25-0014-01 CISC project.} 
\institute{INRIA, Universit\'e C\^ote d'Azur, 
France} 
\email{ilaria.castellani@inria.fr}
\and
Paola Giannini\footnote{This work was partially funded by the MUR project ``T-LADIES'' (PRIN 2020TL3X8X) and has the financial support of the Universit\`a  del Piemonte Orientale.}
\institute{DiSSTE,
Universit\`{a} del Piemonte Orientale, Italy}
\email{paola.giannini@uniupo.it}
}

\def\titlerunning{Towards a semantic characterisation of global type well-formedness}

\def\authorrunning{I. Castellani
  and P. Giannini}

\begin{document}

\maketitle


\begin{abstract}
  We address the question of characterising the well-formedness
  properties of multiparty session types semantically, i.e., as
  properties of the semantic model used to interpret types.  Choosing
  Prime Event Structures (PESs) as our semantic model, we present
  semantic counterparts for the two properties that underpin global
  type well-formedness, namely projectability and boundedness, in this
  model. As a first step towards a characterisation of the class of
  PESs corresponding to well-formed global types, we identify some
  simple structural properties satisfied by such PESs.
  
 \end{abstract}


\section{Introduction} \mylabel{sec:intro} This paper builds on our
previous work~\cite{CDG23}, where we investigated the use of Event
Structures (ESs) as a denotational model for multiparty session types
(MPSTs).  That paper presented an ES semantics
for both sessions and global types,  using respectively Flow
Event Structures (FESs) and Prime Event Structures (PESs),  and
showed that if a session is typable with a given global type, then
 the
FES associated with the session and the PES associated with the global
type yield isomorphic domains of configurations. 

The ES semantics proposed in~\cite{CDG23} abstracts away from the
syntax of global types, by making explicit the concurrency relation
between independent communications. This abstraction is expected since
ESs are a ``true concurrency'' model, where concurrency is treated as
a primitive notion. However,~\cite{CDG23} focussed on the
equivalence between the FES of a session and the PES of its global type,
without  drawing all the consequences of its results and 
demonstrating the full benefits of the ES semantics for MPSTs.

In the present paper, we move one step further by studying how the
well-formedness property of global types considered in~\cite{CDG23} is
reflected in their interpretation as Prime Event Structures (PESs).
In~\cite{CDG23}, global type well-formedness is the conjunction of a
\emph{projectability} condition and a \emph{boundedness} condition.
Having semantic counterparts for these conditions will enable us to
reason directly on PESs, taking advantage of their faithful account of
concurrency and of their graphical representation.

We prove that: 1) all global types that type the same network yield
identical PESs, 2) our proposed properties of \emph{semantic
projectability} and \emph{semantic boundedness} for PESs reflect the
corresponding properties of global types, and 3) PESs obtained from
global types enjoy some 
 simple  structural properties.

The rest of the
paper is organised as follows. In~\refToSection{sec:syntax}
and~\refToSection{sec:events} we recall the necessary background
from~\cite{CDG23} and present the result 1) above.
In~\refToSection{sec:semantic-wf} we define our semantic notions of
projectability and boundedness and prove the result 2) above.
\refToSection{sec:struct-properties} is devoted to the result 3).
Finally, in~\refToSection{sec:rw-conclusion}
we discuss related work and sketch some directions for future work.


In the paper, all theorems are given with proofs while all lemmas are stated
without proofs.


\section{Networks and Global Types}
\mylabel{sec:syntax}
To set up the stage for our study, we recall the definitions of sessions and global
types from \cite{CDG23}.
In the core multiparty session
calculus of \cite{CDG23},
sessions are described as
networks of sequential processes, and processes
coincide with local types.
Session \emph{participants} are denoted by $\pp,\q,\pr$, and
\emph{messages} by $\la,\la'$.
%

\begin{definition}[Processes and networks]\mylabel{pN}\
\begin{itemize}
 \item {\it Processes} are  defined by: $\PP\coDefGr\oup\pp{i}{I}{\la}{\PP}~~\mid~~ \inp\pp{i}{I}{\la}{\PP}~~\mid~~\inact$ \\
where $I$ is a finite non-empty index set and $\la_h\not=\la_k\,$ for 
$h\neq k$.
\item {\em Networks} are defined by: $\Nt =  \pP{\pp_1}{\PP}\parN
  \dots  \parN  \pP{\pp_n}{\PP}$ with $n\geq 1$ and
  $\pp_h\not=\pp_k\,$ for $h\neq k$.
\end{itemize}
\end{definition}
The symbol $ \coDefGr$ in the definition of processes indicates that
the definition is coinductive. This allows 
 infinite processes to be defined without using an explicit
recursion operator.  However, in order to achieve decidability we
focus on regular processes, namely those with a finite number of
distinct subprocesses.   In writing processes, we will omit
trailing $\inact$'s and when $\cardin{I} = 1$ we will omit the choice
symbol. 

A network is a parallel composition of processes, each located
at a different participant. The LTS semantics of networks is specified by the unique rule:
\begin{figure}[h]
 \centerline{$
\begin{array}{c}
\pP{\pp}{\oup\q{i}{I}{\la}{\PP}}\parN \pP{\q}{\inp\pp{j}{J}{\la}{\Q}}\parN\Nt\stackred{\Comm\pp{\la_k}\q}
  \pP{\pp}{\PP_k}\parN \pP{\q}{\Q_k}\parN\Nt~~~\text{where }
   k \in I{\cap}J{~~~~~~\rulename{Comm}}
\end{array}
$}  
\end{figure}
\begin{definition}[Global types]
\mylabel{def:GlobalTypes} 
Global types $\G$  are defined by: 
$\G \coDefGr  \gt\pp\q i I \la \G ~\mid ~\End
$\\
where $I$ is finite non-empty index set and $\la_h\not=\la_k\,$ for
$h\neq k$.
\end{definition}
Here again, $\coDefGr$ indicates that the definition is coinductive,
and we focus on \emph{regular} global types.   We will omit
trailing $\End$'s and when $\cardin{I} = 1$ we will write a global
type $\pp\to\q : \set{\la; \G}$ simply as $\gtCom\pp\q{\la};
\G$. 

A \emph{communication} $\Comm{\pp}{\M}{\q}$ represents the
transmission of label $\M$ on the channel $\pp\q$ from $\pp$ to $\q$.
Communications are ranged over by $\alpha, \beta$. The following
notion of trace will be extensively used in the sequel. 



%
\begin{definition}[Traces]
A \emph{trace} $\comseq, \tau$ is a
finite sequence of communications, i.e.
$\comseq::=\ee\mid\concat\alpha\comseq$.\\
The set of traces is denoted by \Comseq.
\end{definition}
It is useful to define sets of participants also for
communications and traces.  We define
$\participant{\Comm{\pp}{\M}{\q}}=\set{\pp,\q}$, and we lift this
definition
to traces by letting $\participant{\ee} = \emptyset$ and
$\participant{\concat\alpha\comseq} =
\participant{\alpha}\cup\participant{\comseq}$.  

As observed in~\cite{CDG23}, global types may be viewed as trees
whose internal nodes are decorated by
channels $\pp\q$, leaves by $\End$, and edges by labels $\la$.
Given a global type, the sequences of decorations of nodes and edges on
the path from the root to an edge in the tree of the global type are
traces.
We denote by $\FPaths{\G}$ the set of
traces of $\G$. By definition, $\FPaths{\End} = \emptyset$ and each
trace in $\FPaths{\G}$ is non-empty.

The set of {\em participants of a global type $\GP$},
$\participant{\GP}$, 
is 
the union of the sets of participants of its traces, i.e.
$\participant{\GP} = \bigcup_{\comseq \in
    \FPaths{\G}} \participant{\comseq}$.
The regularity assumption ensures that $\participant{\GP}$ is finite
for any $\GP$.

 The semantics of global types is given by the standard LTS presented
in~\refToFigure{ltgt}, where transitions are labelled by
communications. 

%
\begin{figure}[h]
 \centerline{$
\begin{array}{c}
 \gt\pp\q i I \la \G \stackred{\Comm\pp{\la_j}\q}\G_j~~~~~~j\in I{~~~\rulename{Ecomm}}
\quad
 \prooftree
 \G_i\stackred\alpha\G_i' \quad 
 \text{ for all }
i \in I \quad\participant{\alpha}\cap\set{\pp,\q}=\emptyset
 \justifies
 \gt\pp\q i I \la \G \stackred\alpha\gt\pp\q i I \la {\G'} 
 \using ~~~\rulename{Icomm}
  \endprooftree\\ \\
\end{array}
$}
\vspace{-10pt}
\caption{
LTS for global types.
}\mylabel{ltgt}
\end{figure}
%

\begin{figure}[t]
 \centerline{
 $
 \begin{array}{c}
 \\[-1pt]
\proj\G{\pr} = \inact \text{ if }\pr\not\in\participant\G \qquad
\proj{(\gt\pp\q i I \M \RG)}\pr=\begin{cases}
  \inpP\pp{i}{I}{\M}{\proj{\RG_i}\pr}    & \text{if }\pr=\q, \\
    \oupP\q{i}{I}{\M}{\proj{\RG_i}\pr}    & \text{if }\pr=\pp, \\
      \proj{\RG_1}\pr  & \text{if } \pr\not\in\set{\pp, \q}\text{ and } \pr\in\participant{\RG_1}\\
      & \text{ and }\proj{\RG_i}\pr=\proj{\RG_1}\pr\text{ for all } i \in I
\end{cases}\\[13pt]
\end{array}
$
}
\caption{Projection of  global types onto participants.} \mylabel{fig:proj}
\end{figure}
 
 The projection of a global type onto participants is given in
\refToFigure{fig:proj}. As usual, projection is defined only when it
is defined on all participants. Due to the simplicity of our calculus,
the projection of a global type, when defined, is simply a
process. The definition is the standard one
from~\cite{HYC08,HYC16}: the projection of a choice type on the sender
or the receiver yields 
an output choice or an input choice, while its projection on a third
participant is its projection on the continuation of the branch,
which must be equal on all branches. 
Our coinductive definition is
more permissive than the standard one for infinite types
(see~\cite{CDG23}). 

A global type $\G$ is \emph{projectable} if $\proj{\GP}{\pp}$ is defined for
 all $\pp$. 
 
The following property of boundedness for global
types is used to ensure progress.

\begin{definition}[Depth and boundedness]\label{depth}
 The two functions $\depthG{\pp}{\comseq}$ and $\depthG{\pp}{\G}$ are defined by:\\
\centerline{$\begin{array}{l}
 \depthG{\pp}{\comseq} =\begin{cases}
   \cardin{\comseq_1\cdot\alpha}  &\text{ if }\comseq =
\concat{\concat{\comseq_1}{\alpha}}{\comseq_2}\text { and }\pp \notin
\participant{\comseq_1}\text { and }\pp \in \participant{\alpha}\\
  0   & \text{otherwise }
\end{cases}
\\
\depthG{\pp}{\G}=
    \sup  (\{\depthG{\pp}{\comseq}\ |\ \comseq\in\FPaths{\G}\})
  \end{array}
$}
A global type $\GP$ is {\em bounded} if $\depthG{\pp}{\G'}$ is finite
for each participant $\pp$ and each subtree $\G'$ of $\GP$.
 \end{definition}

 
If $\depthG{\pp}{\G}$ is finite, then there is no path in
the tree of $\G$ in which $\pp$ is delayed indefinitely. 
Note that if  
$\depthG{\pp}{\G}$ is finite,
$\G$ may have subtrees $\G'$ for which $\depthG{\pp}{\G'}$ is infinite. 
 \begin{definition} [Well-formed global types] \label{wfs} A global
   type $\GP$ is {\em well formed} if 
   it is projectable and bounded.
 \end{definition}



We conclude this section by recalling the type system for
networks. The unique typing rule for networks is
Rule~\rulename{Net} in \refToFigure{fig:typing}.  It relies on a
preorder on processes, $\PP\leq\Q$, 
meaning that {\em process
  $\PP$ can be used where we expect process
  $\Q$}.  This preorder plays the same role as the standard subtyping
for local types, except that it is invariant for output processes
(rather than covariant). This restriction is imposed in~\cite{CDG23}
in order to obtain bisimilar LTSs for networks and their global types,
a property which in turn is used to prove our main result there
(isomorphism of the configuration domains of the two ESs). 
The preorder rules are interpreted coinductively, since processes may have
infinite (regular) trees.

A network is well typed if 
all its participants
behave as specified by the projections of the same global type
$\G$. Rule \rulename{Net} is standard for MPSTs, so we do not discuss it
further.

%
\begin{figure}[h]
{\small \centerline{$
 \begin{array}{c}
\inact\subt\inact~\rulename{ $\subt$ -$\inact$}\quad \cSinferrule{\PP_i\subt\Q_i ~~~~~i\in I}{\inp\pp{i}{I\cup J}{\M}{\PP}\subt \inp\pp{i}{I}{\M}{\Q}}{\rulename{ $\subt$-In}}
\quad \cSinferrule{\PP_i\subt\Q_i ~~~~~i\in
  I}{\oupTx\pp{i}{I}{\M}{\PP}\subt \oup\pp{i}{I}{\M}{\Q}}{\rulename{
    $\subt$-Out}}\\ \\
\inferrule{\PP_i\subt\proj\GP{\pp_i}~~~~~i\in I~~~~~~\participant\GP\subseteq \set{\pp_i\mid i\in I}}
{\derN{\PiB_{i\in I}\pP{\pp_i}{\PP_i}}\GP} ~\rulename{Net}
\\[13pt]
\end{array}
$}}
\caption{Preorder on processes and  network typing rule.} \mylabel{fig:typing}
\end{figure}
%


\section{Event Structure Semantics of Global Types}
\label{sec:events}
In this section we present the interpretation of global types as Prime
Event Structures, as proposed in our previous work~\cite{CDG23}.  We
start by recalling the definition of \emph{Prime Event Structure}
(PES)  and configuration  from~\cite{NPW81}.  All the following definitions (from
\refToDef{def:permEq} to \refToDef{eg}) are required background taken
from~\cite{CDG23}, with some minor variations.
The new material starts immediately after \refToDef{eg}. 
\begin{definition}[~\cite{NPW81} Prime Event Structure] \label{pes} A
  prime event structure {\rm (PES)} is a tuple $S=(E,\leq, \gr)$ where
  $E$ is a denumerable set of events; $\leq\,\,\subseteq (E\times E)$
  is a partial order relation, called the \emph{causality} relation;
  $\gr\subseteq (E\times E)$ is an irreflexive symmetric relation,
  called the \emph{conflict} relation, satisfying the property of
  \emph{conflict hereditariness}:
  $\forall e, e', e''\in E: e \grr e'\leq e''\Rightarrow e \grr e''$.
\end{definition}

A PES configuration is a set of events that may
have occurred at some stage of computation. 
\begin{definition}[~\cite{NPW81} PES configuration] \label{configP}
  Let $S=(E,\leq, \gr)$ be a prime event structure. A
  configuration of $S$ is a finite subset $\ESet$ of $E$ which is
  (1) downward-closed: \ $e'\leq e \in \ESet \, \ \impl \ \ e'\in
  \ESet$; and (2) conflict-free: $\forall e, e' \in \ESet, \neg (e \gr
e')$.
\end{definition}
The semantics of a PES $S$ is given by its poset of configurations
ordered by set inclusion, where $\ESet_1 \subset \ESet_2$ means that
$S$ may evolve from $\ESet_1$ to $\ESet_2$. 

The events of the PES associated with a global type
will be equivalence classes of particular traces.
%
We introduce some notations for traces $\comseq$. 
%
We denote by $\at{\comseq}{i}$ the $i$-th element of $\comseq$.  If
$i \leq j$, we define
$\range{\comseq}{i}{j} = \at{\comseq}{i} \cdots \at{\comseq}{j}$ to be
the subtrace of $\comseq$ consisting of the $(j-i+1)$ elements
starting from the $i$-th one and ending with the $j$-th one.  If
$i > j$, we convene that $\range{\comseq}{i}{j}$ denotes the empty
trace $\ee$.

  A permutation equivalence on $\Comseq$
  is used to swap
  communications with disjoint participants.
\begin{definition}[Permutation equivalence]\mylabel{def:permEq}
The permutation equivalence on
$\Comseq$ is the least equivalence $\sim$ such that
\\
\centerline{$
\concat{\concat{\concat{\comseq}{\alpha}}{\alpha'}}{\comseq'}\,\sim\,\,
\concat{\concat{\concat{\comseq}{\alpha'}}{\alpha}}{\comseq'}
\quad\text{if}\quad \participant{\alpha}\cap\participant{\alpha'}=\emptyset
$}
We denote by $\eqclass{\comseq}$ the equivalence class of 
$\comseq$, and by $\quotient$ the set of equivalence classes on
$\Comseq$.
\end{definition}
The events of the PES associated with a global type 
%
%
are equivalence classes of particular traces that we call
\emph{pointed}. Intuitively,
a pointed trace ``points to'' its last communication,
in that all the preceding communications in the trace
should cause some subsequent communication in the trace. 
Formally:
\begin{definition}[Pointed  trace]\mylabel{pcs}
  A non empty
  trace $\comseq = \range{\comseq}{1}{n}$ is said to be  \emph{pointed}
  if  
  \\
  \centerline{
    $ ~~\forall i \,. 1\leq i<n, ~\exists j \,. \, i < j \leq n \, .
    ~~ \participant{\at{\comseq}i}\cap
    \participant{\at{\comseq}j}\not=\emptyset $}
\end{definition}
 Note that the condition of \refToDef{pcs} is vacuously satisfied
by any trace of length $n =1$, since in that case there is no $i$
such that $1\leq i<n$. 

Let us also point out that \refToDef{pcs} is slightly different from (but equivalent to) the
definition of pointed trace given in~\cite{CDG23}.

For example,
  let $\alpha_1= \Comm\pp{\la_1}\q, \,\alpha_2= \Comm\pr{\la_2}\ps$ and
  $\alpha_3= \Comm\pr{\la_3}\pp$.  Then $\sigma_1 = \alpha_1$ 
  and $\sigma_3 =
  \concat{\concat{\alpha_1}{\alpha_2}}{\alpha_3}\,$ are pointed
   traces,  while  $\sigma_2 = \concat{\alpha_1}{\alpha_2}\,$ is
  \textit{not} a pointed trace.

  If $\comseq$ is non empty, we use $\last{\comseq}$ to denote the
  last communication of $\comseq$. It is easy to prove that, if
  $\comseq$ is a pointed trace and $\comseq \sim \comseq'$, then
  $\comseq'$ is a pointed trace and $\last{\comseq}=\last{\comseq'}$.
\begin{definition}[Global event]\mylabel{def:glEvent}
  Let $\comseq = \concat{\comseq'}{\alpha}\,$ be a pointed trace.
  Then $\gamma = \eqclass{\comseq}$ is a \emph{global event}, also
  called \emph{g-event}, with communication $\alpha$, notation
  $\comm\gamma=\alpha$. 
 \end{definition}
 Notice that, due to the observation above, 
 ${\sf cm}(\gamma)$ is well defined.  We denote by $\GEs$ the set of
 g-events.
 
 We now introduce an operator that adds a communication $\alpha$ in
 front of a g-event $\comocc$, provided $\alpha$ is a cause of
 some communication in the trace of $\comocc$.
  This ensures that the
 operator always transforms a g-event into another g-event.  We call
 this operator \emph{causal prefixing of a g-event by a
   communication}. 
  \begin{definition}[Causal prefixing of a g-event by a communication]\text{~}\\[-15pt]
  \mylabel{causal-path} 
    \begin{enumerate}
    \item\mylabel{causal-path1} The  {\em causal prefixing}
      of a g-event $\gamma\,$ by a communication $\alpha$
       is defined by:\\
     \centerline{$
      \cau{\alpha}{\comocc}=\begin{cases}
          \eqclass{\concat{\alpha}{\comseq}} & \text{if~ $\comocc
            =\eqclass{\comseq}$~ and~
            $\participant{\alpha}\cap\participant{\comseq} \neq \emptyset$}\\
    \comocc  & \text{otherwise}
\end{cases}
$}
\item\mylabel{causal-path2} The operator $\circ$ naturally extends to  
 traces by:  
$\cau{\emptyseq}{\comocc} = \comocc \quad
\cau{(\concat\alpha\comseq)}\comocc=\cau\alpha{(\cau\comseq\comocc)}$
  \end{enumerate}
\end{definition}
An easy consequence of Clause~\ref{causal-path2} is that
$\cau{(\concat{\comseq'}\comseq)}\comocc=\cau{\comseq'}{(\cau\comseq\comocc)}$
for all $\comseq$ and $\comseq'$.

Using causal prefixing, we can define the mapping ${\sf ev}(\cdot)$  which,
 applied to a trace $\comseq$, yields the g-event representing
the communication $\last\comseq$ prefixed by its causes occurring in
$\comseq$. 
\begin{definition}
\mylabel{event-of-sigma}    The {\em g-event generated by a 
  non-empty  trace  } 
 is defined by: $
 \ev{\concat\comseq\alpha}=\cau\comseq{\eqclass\alpha}
$
 \end{definition}
 Clearly, $\ev{\comseq}$ is a subtrace of $\comseq$ and 
 $\comm{\ev\comseq}=\last\comseq$.
  Observe that the function $\ev{\cdot}$ is not injective
 on the set of traces of a global type. For example, let
  $\G =
 \gtBin{\pp}{\q}{\Seq{\la_1}{\gtCom\pr\ps{\la}}}{\Seq{\la_2}{\gtCom\pr\ps{\la}}}$. Let
  $\comseq_1 = \pp\q\la_1\cdot \pr\ps\la$ and
  $\comseq_2 = \pp\q \la_2 \cdot \pr\ps\la$. Then
  $\comseq_1, \comseq_2 \in \FPaths{\G}$ and $\ev{\comseq_1} =
  \ev{\comseq_2} = \eqclass{\pr\ps\la}$.


  \begin{lemma}
      \label{sibling-events}
      If $\participant{\alpha_1} = \participant{\alpha_2}$ and
      $\ev{\comseq\cdot\alpha_1} = \eqclass{\comseq'\cdot\alpha_1}$,
      then
      $\ev{\comseq\cdot\alpha_2} = \eqclass{\comseq'\cdot\alpha_2}$.
    \end{lemma}

  %
  %
    
%
      %
      
We proceed now to define the causality and conflict relations on
g-events.

\begin{definition}[Causality and conflict relations on g-events] \label{sgeo}
The {\em causality} relation $\precP$ and the {\em conflict} relation $\gr$ on the set of g-events $\GEs$ are defined by:
\begin{enumerate}
\item\mylabel{sgeo1} $\comocc\precP\comocc'$
~if~$\comocc=\eqclass\comseq$ and $\comocc'=\eqclass{\concat\comseq{\comseq'}}$ for some $\comseq,\comseq'$;\\[-10pt]
\item\mylabel{sgeo2}
$\comocc\,\grr\,\comocc'$~if~$\comocc =
\eqclass{\comseq\cdot\pp\q\la_1\cdot \comseq_1}$~and~$\comocc' =
\eqclass{\comseq\cdot\pp\q\la_2\cdot \comseq_2}$
  for some $\comseq, \comseq_1,
  \comseq_2, \pp,\q,\la_1,\la_2$ where $\la_1\neq\la_2$.
\end{enumerate}
\end{definition}
If
$\comocc=\eqclass{\concat{\concat\comseq{\alpha}}{\concat{\comseq'}{\alpha'}}}$,
then the communication $\alpha$ must be done before the communication
$\alpha'$. This is expressed by the causality
$\eqclass{\concat\comseq{\alpha}}\precP\comocc$.  An example is
$\eqclass{\Comm\pp\la\q}\precP\eqclass{\concat{\Comm\pr{\la'}\ps}{\concat{\Comm\pp\la\q}{\Comm\ps{\la''}\q}}}$. 
As regards the conflict relation, an example is
$\eqclass{\concat{\concat{\Comm\pr\la\ps}{\Comm\pp{\la_1}\q}}{\Comm\q{\la}\pr}}\gr
\eqclass{\concat{\Comm\pp{\la_2}\q}{\Comm\pr{\la}\ps}}$, since
$\concat{\Comm\pp{\la_2}\q}{\Comm\pr{\la}\ps}\sim
\concat{\Comm\pr{\la}\ps}{\Comm\pp{\la_2}\q}$. 
\begin{definition}[~\cite{CDG23} Event structure of a global type]
  \mylabel{eg} The {\em event structure of the global type } $\GP$ is
  the triple $\ESG{\GP} = (\EGG(\GP), \precP_\GP , \grr_\GP)$
 where: 
 $\EGG(\GP) =
\set{ \ev{\comseq}\ |\ \comseq\in\FPaths{\G}}$   and $\precP_\GP$ and
$\gr_\GP$ are the restrictions of $\precP$ and $\gr$ to
$\EGG(\GP)$.
\end{definition} 

When clear from the context, we shall omit the subscript
$\G$ in the relations $\precP_\G$ and $\gr_\G$. 

In the sequel, a PES obtained
from a global type by~\refToDef{eg} will often be called a g-PES. 

It should be stressed that \refToDef{eg} only makes sense
for global types that are projectable. Such global
types are guaranteed to be \emph{realisable} by
some distributed implementation, i.e., to type some network.
For these global types it has been shown in~\cite{CDG23} that
the semantics in~\refToDef{eg} preserves and reflects the operational semantics,
namely that $\G$ performs a transition sequence labelled by a trace $\sigma$
in the LTS of~\refToFigure{ltgt} if and only if the associated PES admits the
configuration
$X = \set{\ev{\comseq'} \  | \ \comseq' \sqsubseteq \comseq}$.


  \begin{example}\label{non-realisable-global-type}
The global type    $\G=\pp\to\q : \set{\la_1; \gtCom\pr\ps{\la_3}, \la_2; \gtCom\pr\ps{\la_3}}$ 
  is projectable, with projections:
  \[\proj{\G}{\pp}=\PP=\sendL{\q}{\la_1} \oplus\sendL{\q}{\la_2} \qquad
\proj{\G}{\q}=\Q=\rcvL{\pp}{\la_1} + \rcvL{\pp}{\la_2} \qquad
\proj{\G}{\pr}= \R =\sendL{\ps}{\la_3} \qquad
\proj{\G}{\ps}= U = \rcvL{\pr}{\la_3}
  \]
  Clearly, $\G$ types the network
  $\pP{\pp}{\PP}\parN \pP{\q}{\Q} \parN \pP{\pr}{\R} \parN
  \pP{\ps}{U}$.  The PES $S$ associated with $\G$ by~\refToDef{eg} has
  three events
  $\comocc_1= \eqclass{\pp\q\la_1}, \comocc_2 = \eqclass{\pp\q\la_2},
  \comocc_3= \eqclass{\pr\ps\la_3}$, with $\precP_\G = Id$ and
  $\comocc_1 \gr_\G \comocc_2$.

  Consider now the global type
 $\G' =\pp\to\q : \set{\la_1; \End, \la_2; \gtCom\pr\ps{\la_3}}$, 
      where the first branch of the choice
      has no continuation. Clearly, $\G'$ 
    is \emph{not} projectable. However, \refToDef{eg} associates the
   same PES $S$ with $\G'$, whereas $\G'$ and
   $S$ do not have the same operational semantics,
   since $\G' \stackred{\Comm\pp{\la_1}\q}\End$ while
   in the PES $S$ the configuration $\ESet = \set{\comocc_1}$ can be extended
   to the configuration $\ESet' = \set{\comocc_1, \comocc_3}$.
\end{example}


 Our PES interpretation of global types 
explicitly brings out the
concurrency between communications that is left implicit in the syntax
of global types. 
 We prove now  that all well-formed global
types that type the same network yield the same PES 
(\refToTheorem{unique-PES-interpretation}). 
We start by  
proving a weaker theorem, 
which follows from
   results established in~\cite{NPW81}
   and~\cite{CDG23}. 
   We say that two well-formed global types $\G$
   and $\G'$ are \emph{equivalent} if $\derN\Nt\G$ and $\derN\Nt\G'$
   for some network $\Nt$.
   Let $\cong$ denote PES
   isomorphism.
   \begin{theorem}  [ Equivalent well-formed
   global types yield isomorphic PESs]
   \mylabel{isomorphic-PES-interpretation}
   Let  $\G$ and $\G'$ be well-formed global types.
   If $\derN\Nt\G$ and $\derN\Nt\G'$ for some network $\Nt$, then $\ESG{\G}
   \cong \ESG{\G'}$.
  \end{theorem}
  \begin{proof}
     It was shown in~\cite{CDG23} (Theorem 8.18 p 25) that if
    $\derN\Nt\G$ then the domain of configurations of $\ESG{\G}$ is
    isomorphic to the domain of configurations of the  Flow Event
    Structure  associated with $\Nt$.
    Then, from $\derN\Nt\G$ and $\derN\Nt\G'$ it follows that
    $\ESG{\G}$ and $\ESG{\G'}$ have isomorphic domains of
    configurations.
    A classical result in~\cite{NPW81} (Theorem 9 p.~102) establishes
    that a PES $S$ is isomorphic to the PES whose events are the prime
    elements of the domain of configurations of $S$, with 
    causality relation given by set inclusion and 
    conflict relation given by set inconsistency.
    %
    We conclude that $\ESG{\G} \cong \ESG{\G'}$.  
    %
%
%
\end{proof}

We now wish to go
a step further by showing that
$\ESG{\G}$ and $\ESG{\G'}$ are actually \emph{identical} PESs. 
We write $\G\stackred{\sigma}\G'$ if $\G \stackred{\alpha_1}\G_1 \cdots
\stackred{\alpha_n}\G_n$ where $\alpha_1\cdots\alpha_n = \sigma$ and
$\G_n=G'$, and $\G\stackred{\sigma}$ 
if there exists $\G'$
such that $\G\stackred{\sigma}\G'$.
A similar notation will be used for transition sequences of a network $\Nt$. 

%
 Our next theorem relies on the following key lemma. 

\begin{lemma}
  \mylabel{traces-vs-events} Let $\G$ be a global type and
  $\comseq$ be a
  trace. Then:
  \begin{enumerate}
    \item \label{p1} $\comseq\in\FPaths{\G}$ implies $\G\stackred{\sigma}\,$;
    \item \label{p2} $\G\stackred{\sigma}$ 
      implies $\ev{\comseq}\in
\EGG(\G)$.
     \end{enumerate} 
    \end{lemma}
%
%
    \begin{theorem} [ Equivalent well-formed
   global types yield identical PESs]
   \mylabel{unique-PES-interpretation}\  \
   Let $\G$ and $\G'$ be 
 {well}-formed global types.
   If $\derN\Nt\G$ and $\derN\Nt\G'$ for some network $\Nt$, then $\ESG{\G}
   = \ESG{\G'}$.
  \end{theorem}

\begin{proof}
  It is enough to show that $\ESG{\G}$ and $\ESG{\G'}$ have exactly
  the same sets of events, 
  since events are defined syntactically and the relations of
  causality and conflict can be extracted from their syntax.

    We prove that $\EGG(\G) = \EGG(\G')$.  Let
    $e \in \EGG(\G)$. 
    By \refToDef{eg}  there exists
    $\comseq \in\FPaths{\G}$  such that  $\ev{\comseq} = e$. 
           Then  $\G\stackred{\comseq}\,$ 
    by~\refToLemma{traces-vs-events}(\ref{p1}), from which we deduce
     $\Nt\stackred{\comseq}\,$  by the Session Fidelity result in
    \cite{CDG23} (Theorem 6.11 p.~16). Then 
    $\G'\stackred{\comseq}\,$ 
  by
    the Subject Reduction result in \cite{CDG23} (Theorem 6.10
    p.~16).
    By~\refToLemma{traces-vs-events}(\ref{p2}) this implies 
    $\ev{\comseq}\in \EGG(\G')$, i.e., $e\in \EGG(\G')$. 

\end{proof}


\section{Semantic Well-formedness}
\mylabel{sec:semantic-wf}
We now investigate semantic counterparts for the syntactic
well-formedness property of global types. Recall that type
well-formedness is the conjunction of two properties: projectability
and boundedness.  We start by defining a notion of \emph{semantic
  projectability} for PESs.  We first give some auxiliary definitions.
\begin{definition}\mylabel{def:incf}
  Let $\G$ be a global type and
  $\ESG{\GP} = (\EGG(\GP), \precP, \grr)$.  Two events
  $\comocc_1,\comocc_2\in\EGG(\GP)$ are in {\em initial conflict},
  $\comocc_1 \grin \comocc_2$, if
  $\comocc_1=\eqclass{\concat\comseq{\Comm{\pp}{\M_1}{\q}}}$ and
  $\comocc_2=\eqclass{\concat\comseq{\Comm{\pp}{\M_2}{\q}}}$ for some
  $\comseq, \pp, \q, \M_1, \M_2$ such that $\M_1\neq\M_2$.
\end{definition}

\begin{definition}[Projection of traces on participants]
\mylabel{def:projection}
The {\em projection of
  $\comseq$ on $\pr$}, $\projS\comseq\pr$, is defined by:\\
\centerline{$\projS{\ee}\pr=\ee\quad\quad
\projS{(\Comm\pp\la\q\cdot\comseq)}\pr=\begin{cases}
      \sendL\q\la\cdot\projS{\comseq}\pr & \text{if }\pr=\pp\\
      \rcvL\pp\la\cdot\projS{\comseq}\pr & \text{if }\pr=\q\\
     \projS{\comseq}\pr& \text{if }\pr\not\in\set{\pp,\q}
  \end{cases}
$}
\end{definition}
\begin{definition}[Semantic projectability]
  \label{def:semproj}
  Let $\G$ be a global type and
$\ESG{\GP} = (\EGG(\GP), \precP, \grr)$. We say that $\ESG{\GP}$ is {\em semantically projectable} if
  for all $\comocc_1, \comocc_2 \in \EGG(\GP)$ in initial
   conflict: 
\begin{itemize}
\item[] if there is $\comocc'_1=\eqclass{\concat{\comseq_1}{\alpha_1}}$
  with $\comocc_1\precP_\GP\comocc'_1$ and
  $\pr\in\participant{\alpha_1} \backslash \participant{{\sf cm}(\comocc_1)}$,\\[-10pt]
\item[] then there is $\comocc'_2=\eqclass{\concat{\comseq_2}{\alpha_2}}$
 with  $\comocc_2\precP_\GP\comocc'_2$ and
$\alpha_2 = \alpha_1 $ and  $\projS{\comseq_2}\pr=\projS{\comseq_1}\pr$.
\end{itemize}
\end{definition}
Note that  if   $\G$  is semantically projectable, then also
 any subterm $\GP'$  of $\G$ is 
semantically projectable. 

We now show that, if a global type $\G$ is projectable, then
all initial conflicts between two participants $\pp$ and $\q$ in the event
structure $\ESG{\GP}$ reflect branching points between $\pp$ and $\q$
in the tree of $\G$.
In general, the mapping from branching points in
the tree of $\G$ to initial conflicts in $\ESG{\GP}$ is not injective,
namely, there may be several branching points in the tree of $\G$ that
give rise to the same initial conflict in $\ESG{\GP}$.
\begin{lemma}[Initial conflicts in $\ESG{\GP}$ reflect
  branching points in the tree of $\G$]
  \label{initial-conflicts-and-branching-points}
 Let $\G$ be a global type and
 $\ESG{\GP} = (\EGG(\GP), \precP, \grr)$.  Let
 $\comocc_1, \comocc_2 \in \EGG(\GP)$ be in initial conflict and
 $\comocc_i=\eqclass{\concat\comseq{\alpha_i}}$ for $i\in\set{1,2}$.
 If $\GP$ is projectable then there exists $\comseq'$ such that
 $\comseq'\cdot\alpha_i\in\FPaths{\G}$ and
 $\ev{\comseq'\cdot\alpha_i}=
 \eqclass{\comseq\cdot\alpha_i}$ for $i\in\set{1,2}$.
\end{lemma}

Let $\comseq\in\FPaths{\G}$. We denote by $\G_{\comseq}$ the
{\em subterm of $\G$ after $\comseq$}, 
which is easily defined by
induction on the length of $\comseq$. The converse of
\refToLemma{initial-conflicts-and-branching-points} is immediate,
since any subterm of $\G$ is $\G_{\comseq}$ for some
$\comseq\in\FPaths{\G}$. So if $\G_{\comseq}=\gt\pp\q i I \la {\G'}$
with $\set{1,2}\subseteq I$, then the two events
 $\comocc_1= \ev{\comseq\cdot \Comm{\pp}{\M_1}{\q}}$
and
$\comocc_2= \ev{\comseq\cdot \Comm{\pp}{\M_2}{\q}}$
are in
initial conflict because by~\refToLemma{sibling-events}  there exists $\comseq'$ such
that
$\ev{\comseq\cdot \Comm{\pp}{\M_i}{\q}}=\eqclass{\comseq'\cdot \Comm{\pp}{\M_i}{\q}}$ for
$i\in\set{1,2}$.
\begin{theorem}[Projectability preservation]
  \label{sem-projectability}
If $\GP$ is projectable then $\ESG{\GP}$ is
semantically projectable.
\end{theorem}

\begin{proof} 
  Let $\comocc_1, \comocc_2 \in \EGG(\GP)$ and
  $\comocc_1 \grin \comocc_2$.  By definition,
  $\comocc_1=\eqclass{\concat\comseq{\Comm{\pp}{\M_1}{\q}}}$ and
  $\comocc_2=\eqclass{\concat\comseq{\Comm{\pp}{\M_2}{\q}}}$ for some
  $\comseq, \pp, \q, \M_1, \M_2$ such that $\M_1\neq\M_2$. Let
  $\alpha_i = \Comm{\pp}{\M_i}{\q}$ for $i\in\set{1,2}$. \\
   Since $\GP$ is projectable, by
  \refToLemma{initial-conflicts-and-branching-points} there exists
  $\comseq'$ such that $\comseq'\cdot\alpha_i\in\FPaths{\G}$ and
  $\ev{\comseq'\cdot\alpha_i}=\eqclass{\comseq\cdot\alpha_i}$ for
  $i\in\set{1,2}$.  Then $\G_{\comseq'}= \gt\pp\q i I \la {\G'}$ with
  $\set{1,2} \subseteq I$. Since $\G$ is projectable, also
  $\G_{\comseq'}$ is projectable.
  Thus, for any $\pr \not\in\set{\pp,\q}$ we get $\proj{\G'_1}{\pr}=\proj{\G'_2}{\pr}$. \\
  Let $\comocc'_1\in \EGG(\GP)$, with $\comocc_1\precP_\GP\comocc'_1$,
  ${\sf cm}(\comocc'_1) = \beta$, and
  $\pr\in\participant{\beta}\backslash \set{\pp, \q}$. Since
  $\comocc_1\precP_\GP\comocc'_1$,
  it must necessarily be
   $\comocc'_1 =\eqclass{\comseq\cdot\alpha_1\cdot\comseq_1\cdot
     \beta}
   = \ev{\comseq'\cdot\alpha_1\cdot\comseq'_1\cdot
     \beta}$ for some $\comseq_1, \comseq'_1$. Then $\comseq'_1\cdot
   \beta$ is a path in $\G'_1$. 
   \\
        Since $\G_{\comseq'}$ is projectable, $\proj{\G'_1}{\pr}=\proj{\G'_2}{\pr}$. Then there must
      be a path $\comseq'_2$ of $\G'_2$ such that 
      $\projS{\comseq'_1\cdot\beta}{\pr}=\projS{\comseq'_2\cdot\beta}{\pr}$. We
      want to show that $\comocc'_2 =
      \ev{\concat{\comseq'\cdot\alpha_2\cdot\comseq'_2}{\beta}} =
      \eqclass{\concat{\comseq'\cdot\alpha_2\cdot\comseq_2}{\beta}} $
      for some $\comseq_2$, i.e., that $\comocc_2 \precP \comocc'_2
      $.  Now, if 
       $\participant{\beta} \cap \set{\pp,\q} \neq \emptyset$, 
       we can
       conclude immediately. So, let us assume
       $\participant{\beta}\cap \set{\pp,\q} 
       = \emptyset$. 
       \\
      Since  $\comocc'_1 = \ev{\comseq'\cdot\alpha_1\cdot\comseq'_1\cdot
     \beta} =
       \eqclass{\comseq\cdot\alpha_1\cdot\comseq_1\cdot
         \beta}$ we know that  $\alpha_1\cdot\comseq_1\cdot \beta$ is
       a pointed trace. So, 
        there must be a \emph{bridging communication sequence} between
        $\alpha_1$ and $\beta$,  namely
        there must be a subtrace $\beta_1\cdots \beta_n$ of
        $\comseq_1$ for some $n \geq 1$
        such that 
        
        \centerline{
       $ \participant{\alpha_1}\cap \participant{\beta_1} \neq \emptyset\qquad
\participant{\beta_n}\cap \participant{\beta} \neq \emptyset \qquad
       \participant{\beta_i}\cap \participant{\beta_{i+1}} 
       \neq \emptyset~~~\mbox{for}~~1\leq i <n
   $}

\noindent
Correspondingly, we will have
         $\comseq'_1 =
         \widehat{\comseq}_1\cdot\beta_1\,\cdots\, \widehat{\comseq}_n
         \cdot \beta_{n} $.
         There are now two possible cases:

         - $\participant{\beta_i} \neq \set{\pp, \q} $
           for every $i =1, \ldots, n$.
           Since $\G_{\comseq'}$ is projectable,
           $\proj{\G'_1}{\ps}=\proj{\G'_2}{\ps}$ for all
           $\ps\not\in\set{\pp,\q}$, i.e.,
           $\projS{\comseq'_1}{\ps}=\projS{\comseq'_2}{\ps}$.
           Therefore all the $\beta_{i}$'s for
           $1\leq i \leq n$ must occur in the
           same order  in $\comseq'_2$, i.e.
           $\comseq'_2 = \tau_1\cdot\beta_1\,\cdots\, \tau_n \cdot
           \beta_n $ for some
           $\tau_1, \ldots, \tau_n$.
Hence 
$\ev{\concat{\comseq'\cdot\alpha_2\cdot\comseq'_2}{\beta}} =
\eqclass{\concat{\comseq'\cdot\alpha_2\cdot\comseq_2}{\beta}} $ for
some $\comseq_2$. 
           
-  $\participant{\beta_j} = \set{\pp, \q} $ for some $j$,
  $1\leq j \leq n$. Let $k$ be the maximum such index $j$. Then we know
  that $\participant{\beta_{h}}\neq\set{\pp,\q}$ for every $h, k+1\leq h \leq n$,
  and either $\pp\in\participant{\beta_{k+1}}$ or
  $\q\in\participant{\beta_{k+1}}$.
Therefore all the $\beta_{h}$'s for $k+1\leq h \leq n$ must
occur in the same order in $\comseq'_2$, i.e.
 $\comseq'_2 = \tau_k\cdot\beta_{k+1}\,\cdots\, \tau_n \cdot
 \beta_n $ for some 
 $\tau_k, \ldots, \tau_n$.
Hence, $\ev{\concat{\comseq'\cdot\alpha_2\cdot\comseq'_2}{\beta}} =
\eqclass{\concat{\comseq'\cdot\alpha_2\cdot\comseq_2}{\beta}} $ for
some $\comseq_2$.
\end{proof}

The converse is not true, i.e., semantic projectability of $\ESG{\GP}$
does not imply projectability of $\G$, as shown by the global type
$\G'$ in \refToExample{non-realisable-global-type}. Note that there is
no network behaving as prescribed by $\G'$. We conjecture that for
realisable global types semantic projectability implies
projectability.

{


%
\mylabel{sec:semantic-bound}


We now define a notion of semantic boundedness for PESs, which is
\emph{global} in that it looks simultaneously at all occurrences of each
participant $\pp$ in the g-events whose last communication involves
$\pp$.
Let $S = (E, \precP, \grr)$ be a g-PES.
For any participant $\pp$, let $\pp \in
\participant{S}$ if there exists $\comocc\in E$ such that
$\pp \in \participant{\comocc}$. 
\begin{definition}[Semantic $k$-depth and semantic boundedness]\label{sem-k-depth}
 Let $S = (E, \precP, \grr)$ be a g-PES.
 The two functions $\depthk{k}{\pp}{\comocc}$ and $\gdepthS{\pp}{S}$ are defined by:\\
 \[\begin{array}{l}
 \depthk{k}{\pp}{\eqclass{\comseq}} =\begin{cases} \cardin{\comseq} &\text{
     if } ~\comseq =
   \comseq_1\cdot\alpha_1 \cdots \comseq_k\cdot\alpha_k
   \begin{array}[t]{l}
     \text {  and  } \pp \in \participant{\alpha_i} ~~\mbox{for $i = 1,\ldots, k$} \\
    \text {  and  } \pp \notin \participant{\comseq_i} ~~\mbox{for $i =
     1,\ldots, k$}
     \end{array}\\
  0   & \text{otherwise }
\end{cases}
\\
\gdepthkS{\pp}{S}  =
    \sup  (  \{\depthk{k}{\pp}{\comocc}\ |\ \comocc
               \in E \}  )   ~~ \mbox{for every}~ k\in \Nat 
   \end{array}
 \]
 $S$ is {\em semantically bounded} if $\gdepthkS{\pp}{S}$ is
 finite for each participant $\pp \in\participant{S}$ and each
 $k\in \Nat$.
 \end{definition}

\begin{theorem}[Boundedness preservation]
   \label{boundedness-characterisation}
   If 
   $\G$ is bounded, then 
   $\ESG{\GP}$ is semantically bounded.
\end{theorem}

\begin{proof}
  Let $\G$ be bounded and $\ESG{\GP} = (\EGG(\GP), \precP,
   \grr)$. We want to show that
   $\gdepthkS{\pp}{\ESG{\GP}}$ is finite for each participant 
   $\pp \in\participant{\ESG{\GP}}$ and each $k\in
   \Nat$.  Fix some $\pp$. If
   $\pp\not\in\participant{\G}$ then $\pp
   \not\in\participant{\ESG{\GP}}$ and thus the statement is vacuously
   true.  So, assume
   $\pp\in\participant{\G}$. 
   We show now, by induction on
   $k$, that there exists $n_k \in \Nat$ such that $\depthk{k}{\pp}{\comocc} \leq
   n_k$ for any $\comocc\in \EGG(\GP)$. 

- Case $k=1$. 
 We may assume $\comocc = \eqclass{\comseq_1\cdot
    \alpha_1} \in
  \EGG(\G)$ with
  $\pp\not\in\participant{\comseq_1}$ and
  $\pp\in\participant{\alpha_1}$, since for any $\comocc'$ not of this shape 
  we have $\depthk{1}{\pp}{\comocc'} = 0$ and we can immediately conclude.
  Then there exists
  $\comseq'_1\cdot\alpha_1\in\FPaths{\G}$ such that $\comocc =
  \ev{\comseq'_1\cdot\alpha_1}$. Note that it must be
  $\pp\not\in\participant{\comseq'_1}$, since otherwise there would
  be some $\beta$ in $\comseq'_1$ such that
  $\participant{\beta}\cap \participant{\alpha_1} \neq
  \emptyset$ and the function $\ev{\cdot}$ would keep this
  $\beta$, contradicting the hypothesis $\pp \not
  \in \participant{\comseq_1}$.
Let $n_1 = \depthG{\pp}{\G} = \sup (\{\depthG{\pp}{\comseq}\ |\
   \comseq\in\FPaths{\G}\})$.
  By~\refToDef{depth} 
  ${\depthG{\pp}{\comseq'_1\cdot\alpha_1}} =  \cardin{\comseq'_1\cdot\alpha_1} \leq
  n_1$. 
  Then by~\refToDef{sem-k-depth} we get 
  $\depthk{1}{\pp}{\eqclass{\comseq_1\cdot\alpha_1}} ={
  \cardin{\comseq_1\cdot\alpha_1} \leq
  \cardin{\comseq'_1\cdot\alpha_1} \leq n_1}$.
  

- Case $k> 1$.  Assume that
  $\sup (\{\depthk{k-1}{\pp}{\comocc}\ |\ \comocc\in \EGG(\GP) \})
  \leq n_{k-1}$.
  Let $\comseq = \comseq_1\cdot \alpha_1\cdots \comseq_k\cdot
    \alpha_k$ be such that
  $\pp\not\in\participant{\comseq_i}$ and
  $\pp\in\participant{\alpha_i}$ for every $i = 1, \ldots, k$, and let
  $\comocc = \eqclass{\comseq}
  \in \EGG(\G)$.
  Then there exists
  $\comseq' = \comseq'_1\cdot\alpha_1 \cdots \comseq'_k\cdot
  \alpha_k\in\FPaths{\G}$ such that $\comocc = \ev{\comseq'}$. 
  For each $i= 1, \ldots, k$ we must have
  $\pp\not\in\participant{\comseq'_i}$, because otherwise we would
  contradict the hypothesis $\pp\not\in\participant{\comseq_i}$ (as
  argued in the previous case).  Let now
  $\comseq'' = \comseq'_1\cdot\alpha_1 \cdots \comseq'_{k-1}\cdot
  \alpha_{k-1}$, and consider the subterm $\G_{\comseq''}$ of $\G$.

  Since $\G$ is bounded and $\G_{\comseq''}$ is a subtree of
  $\G$, by~\refToDef{depth} we get $\depthG{\pp}{\G_{\comseq''}}
  = \sup (\{\depthG{\pp}{\comseq}\ |\
  \comseq\in\FPaths{\G_{\comseq''}}\}) = m\,$ for some
  $m\in\Nat$. Therefore
  $\depthk{1}{\pp}{\eqclass{\comseq_k\cdot\alpha_k}} =
  \cardin{\comseq_k\cdot\alpha_k} \leq
  \cardin{\comseq'_k\cdot\alpha_k} =
  \depthG{\pp}{\comseq'_k\cdot\alpha_k} \leq
  m$.  $\depthG{\pp}{\comseq'_k\cdot\alpha_k} =
  \cardin{\comseq'_k\cdot\alpha_k} \leq m$.
  Let $n_k = n_{k-1} + m$.
  We may conclude that
  $\depthk{k}{\pp}{\eqclass{\comseq_1\cdot \alpha_1\cdots \comseq_k\cdot
      \alpha_k} } =
  \depthk{k-1}{\pp}{\eqclass{\comseq_1\cdot \alpha_1\cdots \comseq_{k-1}\cdot
      \alpha_{k-1}} } +  \depthk{1}{\pp}{\eqclass{\comseq_k\cdot\alpha_k}}  
  \leq n_{k-1} + m
  = n_k$.


\end{proof}

\section{Structural Properties of g-PESs}
\mylabel{sec:struct-properties} In this section we discuss some 
additional 
properties of the PESs we obtain by
interpreting 
global types.  Some of these properties do not depend on the
well-formedness of global types but only on their syntax.
For instance, since we adopt for global types the \emph{directed choice}
construct of~\cite{HYC08,HYC16}: $\gt\pp\q i I \la \G$,
%
%
every branch of a choice uses the same channel $\pp \q$.
%
As a consequence, g-PESs satisfy the property of
\emph{initial conflict uniformity}: in every 
set $X= \set{\comocc_1, \ldots, \comocc_n}$ of initially conflicting
g-events, every $\comocc_i\in X$
uses the same channel $\pp\q$ in its last communication, i.e.,
$\comm{\comocc_i} = \pp\q\la_i$ for some $\la_i$.  Moreover, since
global types have deterministic LTSs, where no state can perform two
different transitions with the same label, the same holds for g-PESs:
if $\ESet, \ESet \cup \set{\comocc_1},
\ESet \cup \set{\comocc_2}$ are configurations of the same g-PES,
then $\comocc_1 \neq \comocc_2$ implies
$\comm{\comocc_1} \neq \comm{\comocc_2}$.
If moreover $\neg(\comocc_1 \gr \comocc_2)$, we additionally have 
$\participant{\comm{\comocc_1}} \cap \participant{\comm{\comocc_2}} =
\emptyset$.


Note that our core session calculus may be viewed as a \emph{linear}
subcalculus of Milner's calculus CCS, where parallel composition
appears only at top level and any pair of processes $P$ and $Q$, run
by participants $\pp$ and $\q$, can communicate only via two
unidirectional channels: channel $\pp\q$ for communication from $\pp$
to $\q$, and channel $\q \pp$ for communication from $\q$ to
$\pp$. Hence, \emph{restricted parallel composition}, which is the
kind of parallel composition used in session calculi, where processes
are only allowed to communicate with each other but not to proceed
independently, becomes an associative operation, while it is not
associative in full CCS (as observed by Milner in his 1980 book,
see~\cite{DBLP:books/sp/Milner80} page 21).
%
%
%

Then, a natural question is: how does our  ES semantics  for
the linear subcalculus of CCS compare to the ES semantics proposed
in~\cite{BC88a,BC88b} for other fragments of CCS?  To carry out this
comparison, we would need to take a more extensional view of g-PESs,
forgetting about the syntactic structure of g-events and retaining
only their last communication. In other words, we should consider
\emph{Labelled PESs}, where events are labelled by communications
$\pp\q\la$ and have no specific structure.   Moreover, some care
should be taken since, unlike our session calculus, our language for
global types is \emph{not} a subcalculus of CCS: indeed, while the
syntax of global types is included in that of CCS with guarded sums,
their semantics is not the same as that of CCS processes, since 
some communications may be performed under guards. 

More in detail, the work~\cite{BC88a} provides a characterisation of
the class of Labelled PESs obtained by interpreting the fragment of
CCS built from actions $a, b, \ldots$ by means of the three
constructors $+, ;, \parallel$, denoting respectively choice,
sequential composition and parallel composition with no communication.
As a matter of fact, \cite{BC88a} uses slightly more relaxed PESs
where conflict is not required to be hereditary - let us call them
r-PESs - and shows that the Labelled r-PESs obtained for that fragment
of CCS are exactly those satisfying two structural properties called
\emph{triangle freeness} and \emph{N-freeness}.  In conjunction,
these two properties express the possibility of extracting a head operator
among $+, ;, \parallel$ from the structure of the r-PES.
We recall from~\cite{BC88a} the definition of these properties. Let $\smile$
denote the \emph{concurrency} relation on the events of a PES
$S=(E,\leq, \gr)$, defined by
$\smile \, = (E\times E) - (\leq \cup\geq\cup\grr)$.  Let
$\diamond = (\leq \cup\geq)$ denote \emph{causal connection}. By
definition the three relations $\diamond, \grr$ and $\smile$ set a
partition over $E\times E$. 
Then triangle freeness (or $\nabla$-freeness) is defined as the
absence of a triple of events $e,e', e''$ such that
$e \,\diamond \, e' \grr \, e'' \smile e$ (see~\cite{BC88a} page 41).
Note that one half of triangle-freeness, where $\diamond$ is replaced
by $\geq$, is implied by conflict hereditariness in g-PESs.
%
We conjecture that g-PESs satisfy also the other half of
triangle-freeness, where $\diamond$ is replaced by $\leq$, namely they
do not feature the pattern 
$e \,\leq \, e' \grr \, e'' \smile e$, a situation known as
\emph{asymmetric confusion} in Petri nets.  The property of N-freeness
is slightly more involved. For any $R \in \set{\leq, \grr, \smile}$,
let $R^\varepsilon$ be the reflexive and symmetric closure of $R$
and $\ddagger(R)$ be the $R-incomparability$ relation defined by
$\ddagger(R) = (E\times E)-R^\varepsilon$. Then the N-freeness
property
is stated as follows:

\centerline{ \emph{N-freeness} \quad
  $~\forall\,R \in \set{\leq, \grr, \smile}$:
  (~$e_0 R e_1 \,\land\, e_0 \ddagger(R) e_2 \,\land\, e_2 R e_3
  \,\land\, e_1 \ddagger(R) e_3\,) \implies (\,e_0 R e_3 \implies e_2
  R e_1\,)$ }  
  This property does not hold for g-PESs, e.g.,
it does not hold for the g-PES of the global type
$\G = \gtCom\pp\q{\la_0} ; \gtCom\pp\pt{\la_1} ; \gtCom\pr\ps{\la_2} ;
\gtCom\q\ps{\la_3}$, with
$e_0 = \eqclass{\pp\q{\la_0}}, e_1 =
\eqclass{\pp\q{\la_0}\cdot\pp\pt{\la_1}}, e_2 =
\eqclass{\pr\ps{\la_2}}, e_3 = \eqclass{\pp\q{\la_0}\cdot
  \pr\ps{\la_2}\cdot \q\ps{\la_3}} $.

However, we may show that g-PESs 
satisfy particular instances of N-freeness, for instance
when $R$ is the covering relation of $\leq$ and
$e_1 \grin e_3$ (in which case conflict hereditariness enforces $e_0 \smile e_2$).
%

The paper~\cite{BC88b}, on the other hand, presents a Flow Event
Structure semantics for the whole calculus CCS. Our conjecture is that
this semantics should coincide with the Flow ES semantics proposed for
 sessions  in~\cite{CDG23}.  However, this is not entirely trivial since
the semantics of~\cite{BC88b} uses self-conflicting events, a specific
feature of Flow ESs, to interpret restricted parallel composition,
while the semantics of~\cite{CDG23} uses a pre-processing phase to
rule out the g-events that do not satisfy a causal well-foundedness
condition, and these events are a superset of those that are
self-conflicting in the semantics of~\cite{BC88b}.  However, one may
already observe that the Flow ESs obtained by interpreting sessions
in~\cite{CDG23} trivially satisfy the axiom $\Delta$ put forward
in~\cite{CZ97}
in order to guarantee that CCS parallel composition is a categorical product.\\

\noindent
%








\section{Conclusion}\mylabel{sec:rw-conclusion}


We conclude by further 
discussing related work
and by sketching some directions for future work.

{\bf Related work.} 
In~\refToSection{sec:struct-properties} we compared our PES semantics
for global types to existing ES semantics for other fragments of
CCS. In that case, the comparison was somewhat hindered by the fact
that the target ESs were not exactly the same (PESs vs r-PESs vs Flow
ESs). We now turn to other proposals of denotational models for
MPSTs. The models that are closest to ours are the \emph{graphical
  choreographies} by Guanciale and Tuosto~\cite{TG18}, the
\emph{choreography automata} by Barbanera, Lanese and
Tuosto~\cite{BLT20}, the \emph{global choreographies} by de'Liguoro,
Melgratti and Tuosto~\cite{DLMT22}, and the \emph{branching pomsets}
by Edixhoven et al.~\cite{EJPC24}.  It should be noted that most of
these works deal with asynchronous communication, so ``events'' (or
communications) are split into send events and receive events. Common
well-formedness conditions proposed in these works are
\emph{well-branchedness}~\cite{BLT20,DLMT22,EJPC24}, which in our case
is enforced by the syntax of global types, and
\emph{well-sequencedness}~\cite{BLT20,DLMT22}, which is automatically
enforced by our PES semantics.  
As regards the use of ESs to model
MPSTs, the paper~\cite{DLMT22} also uses PESs
to model (asynchronous)
choreographies, but it needs an additional type system to obtain
projectability (so, the resulting notion of projectability is not totally semantic).
In~\cite{EJPC24}, asynchronous
choreographies are modelled with branching pomsets, a model
featuring both concurrency and choice, which is
compared \mbox{with various classes of ESs.}
%

{\bf Future work.}
In this paper, we have devised semantic counterparts for the
well-formedness conditions of global types.  However, we have only
gone half the way in establishing a characterisation of the class of
Prime ESs representing well-formed global types.  To achieve such a
characterisation, we should prove the converse of
~\refToTheorem{boundedness-characterisation} and the following weaker
form of the converse of ~\refToTheorem{sem-projectability}:

\textbf{Conjecture}
{\it [Projectability reflection]} Let $\G$ be a global type. If $\ESG{\G}$ is
semantically projectable then there exists a projectable global type $\G'$ such
that $\ESG{\G'} = \ESG{\G}$.

If this conjecture were true, then our PES semantics for global types
would also provide a way to ``sanitise'' ill-formed global types. For
instance, starting from the g-PES $\ESG{\G'}$ of the ill-formed global type
$\G'$ of~\refToExample{non-realisable-global-type}, we would be able
to get back to
the well-formed global type $\G$ of the same example or
to the well-formed global type
$\G'' = \Seq{ \gtCom\pr\ps{\la_3}}{\pp\to\q : \set{\la_1; \End,
    \la_2; \End}}$. 
   Once we achieve a characterisation for this class of g-PESs, the next step would be
   to propose an algorithm to synthesise a well-formed global type (or
   directly a network) from a g-PES of this class.
    A further goal would be to semantically characterise less
   restrictive notions of projection, such as the one proposed
   in~\cite{DenielouY13}.

Since g-PESs are images of regular trees, it would be worth
investigating their connection with Regular Event
Structures~\cite{NT02}.  Moreover, as argued in the previous section,
extensional g-PESs, where g-events have no structure and are labelled
by their last communication, may be viewed as Labelled ESs whose
observable behaviour (the communications) is deterministic. Hence,
such extensional g-PESs could be characterised by the trace language
they recognise\footnote{Here,``trace'' should be intended as a
  Mazurkiewicz trace, namely as an equivalence class of standard
  traces with respect to an independence relation $I$ on the alphabet
  of the language, which in our case is given by
  $\pp\q\la ~I~ \pr\ps\la'$ if
  $\set{\pp,\q} \cap \set{\pr,\ps} =\emptyset$.}.


\paragraph{Acknowledgments.}
We are grateful to the anonymous reviewers for their useful
suggestions. The first author would also like to acknowledge
interesting discussions with Nobuko Yoshida, Francisco Ferreira and
Raymond Hu during her visits to Oxford University and Queen Mary
University of London in 2023.

%

\begin{thebibliography}{10}
\providecommand{\bibitemdeclare}[2]{}
\providecommand{\surnamestart}{}
\providecommand{\surnameend}{}
\providecommand{\urlprefix}{Available at }
\providecommand{\url}[1]{\texttt{#1}}
\providecommand{\href}[2]{\texttt{#2}}
\providecommand{\urlalt}[2]{\href{#1}{#2}}
\providecommand{\doi}[1]{doi:\urlalt{https://doi.org/#1}{#1}}
\providecommand{\eprint}[1]{arXiv:\urlalt{https://arxiv.org/abs/#1}{#1}}
\providecommand{\bibinfo}[2]{#2}

\bibitemdeclare{inproceedings}{BLT20}
\bibitem{BLT20}
\bibinfo{author}{Franco \surnamestart Barbanera\surnameend},
  \bibinfo{author}{Ivan \surnamestart Lanese\surnameend} \&
  \bibinfo{author}{Emilio \surnamestart Tuosto\surnameend}
  (\bibinfo{year}{2020}): \emph{\bibinfo{title}{Choreography Automata}}.
\newblock In \bibinfo{editor}{Simon \surnamestart Bliudze\surnameend} \&
  \bibinfo{editor}{Laura \surnamestart Bocchi\surnameend}, editors: {\slshape
  \bibinfo{booktitle}{Coordination Models and Languages - 22nd {IFIP} {WG} 6.1
  International Conference, {COORDINATION} 2020}}, \bibinfo{volume}{12134},
  \bibinfo{publisher}{Springer}, pp. \bibinfo{pages}{86--106},
  \doi{10.1007/978-3-030-50029-0\_6}.

\bibitemdeclare{article}{BC88a}
\bibitem{BC88a}
\bibinfo{author}{G{\'e}rard \surnamestart Boudol\surnameend} \&
  \bibinfo{author}{Ilaria \surnamestart Castellani\surnameend}
  (\bibinfo{year}{1988}): \emph{\bibinfo{title}{Concurrency and atomicity}}.
\newblock {\slshape \bibinfo{journal}{Theoretical Computer Science}}
  \bibinfo{volume}{59}(\bibinfo{number}{1-2}), pp. \bibinfo{pages}{25--84},
  \doi{10.1016/0304-3975(88)90096-5}.

\bibitemdeclare{inproceedings}{BC88b}
\bibitem{BC88b}
\bibinfo{author}{G{\'e}rard \surnamestart Boudol\surnameend} \&
  \bibinfo{author}{Ilaria \surnamestart Castellani\surnameend}
  (\bibinfo{year}{1988}): \emph{\bibinfo{title}{Permutation of transitions: an
  event structure semantics for {CCS} and {SCCS}}}.
\newblock In \bibinfo{editor}{J.W. \surnamestart de~Bakker\surnameend},
  \bibinfo{editor}{W.-P. \surnamestart de~Roever\surnameend} \&
  \bibinfo{editor}{G.~\surnamestart Rozenberg\surnameend}, editors: {\slshape
  \bibinfo{booktitle}{{REX} School/Workshop on Linear Time, Branching Time and
  Partial Order in Logics and Models for Concurrency, {\rm Noordwijkerhout}}},
  {\slshape \bibinfo{series}{Lecture Notes in Computer Science}}
  \bibinfo{volume}{354}, \bibinfo{publisher}{Springer-Verlag}, pp.
  \bibinfo{pages}{411--427}, \doi{10.1007/BFb0013028}.

\bibitemdeclare{article}{CDG23}
\bibitem{CDG23}
\bibinfo{author}{Ilaria \surnamestart Castellani\surnameend},
  \bibinfo{author}{Mariangiola \surnamestart Dezani{-}Ciancaglini\surnameend}
  \& \bibinfo{author}{Paola \surnamestart Giannini\surnameend}
  (\bibinfo{year}{2023}): \emph{\bibinfo{title}{Event structure semantics for
  multiparty sessions}}.
\newblock {\slshape \bibinfo{journal}{J. Log. Algebraic Methods Program.}}
  \bibinfo{volume}{131}, p. \bibinfo{pages}{100844},
  \doi{10.1016/j.jlamp.2022.100844}.

\bibitemdeclare{article}{CZ97}
\bibitem{CZ97}
\bibinfo{author}{Ilaria \surnamestart Castellani\surnameend} \&
  \bibinfo{author}{Guo~Qiang \surnamestart Zhang\surnameend}
  (\bibinfo{year}{1997}): \emph{\bibinfo{title}{Parallel product of event
  structures}}.
\newblock {\slshape \bibinfo{journal}{Theoretical Computer Science}}
  \bibinfo{volume}{179}(\bibinfo{number}{1-2}), pp. \bibinfo{pages}{203--215},
  \doi{10.1016/S0304-3975(96)00104-1}.

\bibitemdeclare{article}{DLMT22}
\bibitem{DLMT22}
\bibinfo{author}{Ugo \surnamestart de'Liguoro\surnameend},
  \bibinfo{author}{Hern{\'{a}}n~C. \surnamestart Melgratti\surnameend} \&
  \bibinfo{author}{Emilio \surnamestart Tuosto\surnameend}
  (\bibinfo{year}{2022}): \emph{\bibinfo{title}{Towards refinable
  choreographies}}.
\newblock {\slshape \bibinfo{journal}{J. Log. Algebraic Methods Program.}}
  \bibinfo{volume}{127}, p. \bibinfo{pages}{100776},
  \doi{10.1016/j.jlamp.2022.100776}.

\bibitemdeclare{inproceedings}{DenielouY13}
\bibitem{DenielouY13}
\bibinfo{author}{Pierre{-}Malo \surnamestart Deni{\'{e}}lou\surnameend} \&
  \bibinfo{author}{Nobuko \surnamestart Yoshida\surnameend}
  (\bibinfo{year}{2013}): \emph{\bibinfo{title}{Multiparty Compatibility in
  Communicating Automata: Characterisation and Synthesis of Global Session
  Types}}.
\newblock In \bibinfo{editor}{Fedor~V. \surnamestart Fomin\surnameend},
  \bibinfo{editor}{Rusins \surnamestart Freivalds\surnameend},
  \bibinfo{editor}{Marta~Z. \surnamestart Kwiatkowska\surnameend} \&
  \bibinfo{editor}{David \surnamestart Peleg\surnameend}, editors: {\slshape
  \bibinfo{booktitle}{Automata, Languages, and Programming - 40th International
  Colloquium, {ICALP} 2013, Riga, Latvia, July 8-12, 2013, Proceedings, Part
  {II}}}, {\slshape \bibinfo{series}{Lecture Notes in Computer Science}}
  \bibinfo{volume}{7966}, \bibinfo{publisher}{Springer}, pp.
  \bibinfo{pages}{174--186}, \doi{10.1007/978-3-642-39212-2\_18}.

\bibitemdeclare{article}{EJPC24}
\bibitem{EJPC24}
\bibinfo{author}{Luc \surnamestart Edixhoven\surnameend},
  \bibinfo{author}{Sung{-}Shik \surnamestart Jongmans\surnameend},
  \bibinfo{author}{Jos{\'{e}} \surnamestart Proen{\c{c}}a\surnameend} \&
  \bibinfo{author}{Ilaria \surnamestart Castellani\surnameend}
  (\bibinfo{year}{2024}): \emph{\bibinfo{title}{Branching pomsets: Design,
  expressiveness and applications to choreographies}}.
\newblock {\slshape \bibinfo{journal}{J. Log. Algebraic Methods Program.}}
  \bibinfo{volume}{136}, p. \bibinfo{pages}{100919},
  \doi{10.1016/J.JLAMP.2023.100919}.

\bibitemdeclare{inproceedings}{HYC08}
\bibitem{HYC08}
\bibinfo{author}{Kohei \surnamestart Honda\surnameend}, \bibinfo{author}{Nobuko
  \surnamestart Yoshida\surnameend} \& \bibinfo{author}{Marco \surnamestart
  Carbone\surnameend} (\bibinfo{year}{2008}): \emph{\bibinfo{title}{Multiparty
  Asynchronous Session Types}}.
\newblock In \bibinfo{editor}{George~C. \surnamestart Necula\surnameend} \&
  \bibinfo{editor}{Philip \surnamestart Wadler\surnameend}, editors: {\slshape
  \bibinfo{booktitle}{POPL}}, \bibinfo{publisher}{ACM Press},
  \bibinfo{address}{New York}, pp. \bibinfo{pages}{273--284},
  \doi{10.1145/1328897.1328472}.

\bibitemdeclare{article}{HYC16}
\bibitem{HYC16}
\bibinfo{author}{Kohei \surnamestart Honda\surnameend}, \bibinfo{author}{Nobuko
  \surnamestart Yoshida\surnameend} \& \bibinfo{author}{Marco \surnamestart
  Carbone\surnameend} (\bibinfo{year}{2016}): \emph{\bibinfo{title}{Multiparty
  Asynchronous Session Types}}.
\newblock {\slshape \bibinfo{journal}{Journal of {ACM}}}
  \bibinfo{volume}{63}(\bibinfo{number}{1}), pp. \bibinfo{pages}{9:1--9:67},
  \doi{10.1145/2827695}.

\bibitemdeclare{book}{DBLP:books/sp/Milner80}
\bibitem{DBLP:books/sp/Milner80}
\bibinfo{author}{Robin \surnamestart Milner\surnameend} (\bibinfo{year}{1980}):
  \emph{\bibinfo{title}{A Calculus of Communicating Systems}}.
\newblock {\slshape \bibinfo{series}{Lecture Notes in Computer
  Science}}~\bibinfo{volume}{92}, \bibinfo{publisher}{Springer},
  \doi{10.1007/3-540-10235-3}.

\bibitemdeclare{article}{NPW81}
\bibitem{NPW81}
\bibinfo{author}{Mogens \surnamestart Nielsen\surnameend},
  \bibinfo{author}{Gordon \surnamestart Plotkin\surnameend} \&
  \bibinfo{author}{Glynn \surnamestart Winskel\surnameend}
  (\bibinfo{year}{1981}): \emph{\bibinfo{title}{Petri Nets, Event Structures
  and Domains, Part {I}}}.
\newblock {\slshape \bibinfo{journal}{Theoretical Computer Science}}
  \bibinfo{volume}{13}(\bibinfo{number}{1}), pp. \bibinfo{pages}{85--108},
  \doi{10.1016/0304-3975(81)90112-2}.

\bibitemdeclare{inproceedings}{NT02}
\bibitem{NT02}
\bibinfo{author}{Mogens \surnamestart Nielsen\surnameend} \&
  \bibinfo{author}{P.~S. \surnamestart Thiagarajan\surnameend}
  (\bibinfo{year}{2002}): \emph{\bibinfo{title}{Regular Event Structures and
  Finite Petri Nets: The Conflict-Free Case}}.
\newblock In \bibinfo{editor}{Javier \surnamestart Esparza\surnameend} \&
  \bibinfo{editor}{Charles \surnamestart Lakos\surnameend}, editors: {\slshape
  \bibinfo{booktitle}{Applications and Theory of Petri Nets 2002, 23rd
  International Conference, {ICATPN} 2002, Adelaide, Australia, June 24-30,
  2002, Proceedings}}, {\slshape \bibinfo{series}{Lecture Notes in Computer
  Science}} \bibinfo{volume}{2360}, \bibinfo{publisher}{Springer}, pp.
  \bibinfo{pages}{335--351}, \doi{10.1007/3-540-48068-4\_20}.

\bibitemdeclare{article}{TG18}
\bibitem{TG18}
\bibinfo{author}{Emilio \surnamestart Tuosto\surnameend} \&
  \bibinfo{author}{Roberto \surnamestart Guanciale\surnameend}
  (\bibinfo{year}{2018}): \emph{\bibinfo{title}{Semantics of global view of
  choreographies}}.
\newblock {\slshape \bibinfo{journal}{J. Log. Algebraic Methods Program.}}
  \bibinfo{volume}{95}, pp. \bibinfo{pages}{17--40},
  \doi{10.1016/j.jlamp.2017.11.002}.

\end{thebibliography}

\end{document}